\documentclass[11pt]{article}
\usepackage[margin=1in]{geometry}
\usepackage{graphicx}
\usepackage[numbers,sort]{natbib}
\usepackage{footnote,enumerate,amsmath,amssymb,amsfonts,amsthm,subcaption,hyperref}

\newcommand*\samethanks[1][\value{footnote}]{\footnotemark[#1]}
\newcommand{\M}{{\cal M}}
\newcommand{\PP}{{\cal P}}
\newtheorem{theorem}{Theorem}
\newtheorem{lemma}[theorem]{Lemma}
\newtheorem{algorithm}{Algorithm}

\begin{document}
\title{Graph Partitioning Methods for Fast Parallel Quantum Molecular Dynamics}
\author{Hristo N.\ Djidjev\thanks{Los Alamos National Laboratory, Los Alamos, NM 87544.} 
\and Georg Hahn\thanks{Department of Statistics, Columbia University, New York, NY 10027.}
\thanks{Part of work was done when authors were visiting Los Alamos National Laboratory.} 
\and Susan M.\ Mniszewski~\thanks{Los Alamos National Laboratory, Los Alamos, NM 87544.}
\and Christian F.A.\ Negre~\thanks{Los Alamos National Laboratory, Los Alamos, NM 87544.}
\and Anders M.N.\ Niklasson~\thanks{Los Alamos National Laboratory, Los Alamos, NM 87544.} 
\and Vivek B.\ Sardeshmukh\samethanks[4]~\thanks{Department of Computer Science, University of Iowa, Iowa City, IA 52242.}
}
\date{}
\maketitle

\begin{abstract}
We study a graph partitioning problem motivated by the simulation of the 
physical movement of multi-body systems on an atomistic
level, where the forces are calculated from a quantum mechanical description of the electrons.
Several advanced algorithms have been published in the literature for such 
simulations that are based on evaluations of matrix polynomials.
We aim at efficiently parallelizing these computations by using
a special type of graph partitioning.
For this, we represent the zero-nonzero structure of a thresholded matrix as a graph and
partition that graph into several components.
The matrix polynomial is then evaluated for each separate submatrix corresponding to the subgraphs and the
evaluated submatrix polynomials are used to assemble the final 
result for the full matrix polynomial.
The paper provides a rigorous definition as well as a mathematical justification 
of this partitioning problem.
We use several algorithms to compute graph partitions and experimentally 
evaluate
their performance with respect to the quality of the partition obtained with 
each method and the time needed to produce it.
\end{abstract}

\section{Introduction}
\label{section_intro}

Molecular dynamics (MD) simulations study the physical movements of multi-body systems on an atomistic level.
The interatomic movements take place at the femtosecond ($10^{-15}$ second) time scale and are integrated in a simulation
for a pre-specified period of time, typically in the pico to nanosecond ($10^{-12}$ to $10^{-9}$ second) range.
In quantum-based molecular dynamics (QMD) simulations, the interatomic forces are calculated 
from an underlying quantum mechanical description of the electronic structure. 

Several QMD methods are available for a variety of materials systems.
The main computational effort in the density functional based self-consistent tight-binding theory \citep{Elstner1998}, 
one of the most efficient and widely used approaches, is
the diagonalization of a matrix, the so-called \textit{Hamiltonian matrix}.
The Hamiltonian matrix encodes the electronic energy of the system
which is needed in order to construct the \textit{density matrix},
which describes the electronic structure of the system.
The self-consistent construction of the density marix is performed at each time step in a QMD simulation,
prior to each force evaluation. 
Using diagonalization, the construction of the density matrix requires a runtime of $O(N^3)$, where
$N$ is the dimension of the Hamiltonian, and is only suitable for systems of small size.
A number of reduced complexity algorithms have been developed during the
last two decades with a runtime which scales only linearly, $O(N)$, with the system size.

One such recent approach was proposed in \citep{Niklasson2002} and relies on a recursive polynomial expansion representation of the density matrix.
In contrast to diagonalization methods, the sparse-matrix second-order spectral projection (SM-SP2) algorithm scales 
linearly with the system size for non-metallic systems and competes or outperforms 
regular diagonalization schemes in both speed and accuracy using dense or sparse matrices \citep{Mniszewski2015}. 
In this method, the density matrix $D$ is computed from the Hamiltonian $H$ using the formula
\begin{equation}
D=\lim\limits_{i \rightarrow \infty } f_i[f_{i-1}[\dots f_0[X_0]\dots]], \label{eq:expansion}
\end{equation}
where the initial matrix $X_0$ is a linearly modified version of $H$ and $f_i(X_i)$ is a quadratic function
(either $X_i^2$ or $2X_i-X_i^2$, depending on the value of the trace of $X_i$ or $X_{i+1}$).
It is usually sufficient to perform no more than $20-30$ iterations in order to obtain a close approximation of $D$.
In order to reduce the computational complexity, thresholding is applied,
where small nonzero elements of the matrix
(typically between $10^{-5}$ to $10^{-7}$)
are replaced by zeroes.

The computational cost of the SM-SP2 algorithm is dominated by the cost of computing a matrix polynomial $P$,
determined by the cost of squaring sparse matrices.
In order to keep the wall-clock time low for large systems, given the large number ($10^4-10^6$) of time steps  needed in a typical QMD simulation, it is
necessary to parallelize the evaluation of the matrix polynomials.
Several parallel algorithms to achieve linear scaling complexity
in each individual matrix-matrix operation based on thresholded
blocked sparse matrix algebra have recently been proposed
\citep{Bock,Hutter,Mniszewski2015,VandeVondele}.
Our graph-based approach provides an alternative formulation that reduces communication overhead and allows scalable parallelism.

In this paper we focus on the computational aspects of an alternative approach for evaluating the matrix polynomial, denoted G-SP2,
which parallelizes the SP2 algorithm using a decomposition of a graph representation of the density matrix into partitions.
(See \citep{Niklasson2016} for a discussion of the physics aspects of the approach.)
We represent the Hamiltonian (or density) matrix as a graph
that models the zero/nonzero structure of the matrix.
Our approach works by
dividing the graph into parts/partitions in such a way that a suitable cost function is minimized.

For this approach to produce accurate results, it is necessary to consider partitions with overlapping parts (or \textit{halos}).  
Clearly, the greater the overlap, the greater the computational overhead will be.
In order to minimize the overall 
computational cost, the aim of this work is thus to investigate partitioning schemes which,
in contrast to traditional approaches minimizing edge cuts,
attempt to minimize the cost of the corresponding polynomial evaluation, which is related to the sizes of the halos.
We will formalize those objectives as a graph partitioning problem and experimentally study several algorithms for its solution.

The goal of our partitioning approach is to enable fast calculation of the density matrix.
Thresholding of the matrix elements as mentioned above is required, justified by the physics of the underlying application;
without it the resulting matrix would quickly become dense due to fill-in.
Communication between processors can be avoided after each iteration in \eqref{eq:expansion} until the entire polynomial is evaluated.
Our approach achieves this: the initial matrix is partitioned and distributed across processors,
then each processor computes independently the polynomials on its assigned submatrices,
and the resulting submatrices are used to assemble the final output.

\section{Partitioned Evaluation of Matrix Polynomials}
\label{section_theory}

In this section we propose an algorithm for a partitioned evaluation of a matrix polynomial and define the cost function for the corresponding graph partitioning problem. (The appendix contains more details and the proofs.)

Let, for any symmetric matrix $X=\{x_{ij}\}$, $G(X)$ denote the graph, called \textit{sparsity graph} of $X$, that encodes the zero-nonzero structure of $X$. Specifically, for the $i$-th row (column) of $X$ there is a vertex $i$ in $G(X)$, and there is an edge between vertices $i$ and $j$ if and only if $x_{ij}\neq 0$.  
Let $A$ be a symmetric $n \times n$ matrix.
We define a generalization of a matrix polynomial in formula \eqref{eq:expansion}.
We define a \textit{thresholded matrix polynomial} of degree $m=2^s$ to be a superposition of operators of the type
\begin{equation}\label{eq:poly2}
P=P_1 \circ T_1 \circ\dots \circ P_s \circ T_s,
\end{equation}
where $P_i$ is a polynomial of degree 2 and $T_i$ is a thresholding operation. Formally, $T_i$ is a graph operator associated with a set of edges $E(T_i)$ such that, for any graph $I$, $T_i(I)$ is a graph with a vertex set $V(I)$ and an edge set $E(I)\setminus T_i$. 

Denote by $P(A)$ the application of a superpositioned operator $P$ as defined above,
consisting of polynomials $P_i$ and thresholding operations $T_i$,
to a matrix $A$ of appropriate dimension.
In the motivating SM-SP2 application, $A$ corresponds to the Hamiltonian and $P(A)$ corresponds to the density matrix.

Let $\PP(G)$ describe the worst-case zero-nonzero structure of $P(A)$ that ignores the possibility of coincidental zeros resulting from cancellation (adding opposite-sign numbers). We assume that all diagonal elements of $A$ are non-zero and that no $E(T_i)$ contains a loop edge.

Let $\Pi=\{\Pi_1,\dots,\Pi_q\}$, where $\Pi_i$ is a union of vertex sets $U_i$ called \textit{core} and $W_i$ called \textit{halo}. We call $\Pi$ a \textit{CH-partition} (or core-halo partition), if the following conditions are satisfied:
\begin{enumerate}[(i)]
\item $\bigcup_iU_i=V(G),\;U_i\cap U_j=\emptyset \mbox{ if $i\neq j$}$;
\item $W_i$ is the set of all neighbors of vertices in $U_i$ that are not in $U_i$.
\end{enumerate}

Let $H=\PP(G)$ and denote by $H_{U_i}$ the subgraph of $H$ induced by all neighbors of $U_i$ in $H$. Denote by $A_{U_i}$ the submatrix of $A$ consisting of all rows and columns that correspond to vertices of $V(H_{U_i})$.
The following main result of this section shows that $P(A)$ can be computed on submatrices of
the Hamiltonian
and hence justifies the parallelized evaluation of a matrix polynomial. One can prove the following.

\begin{lemma}\label{lem:matrixValuesMulti2}
For any $v\in U_i$ and any neighbor $w$ of $v$ in $\PP(G)$, the element of $P(A)$ corresponding to edge $(v,w)$ of $\PP(G)$ is equal to the element of $P(A_{U_i})$ corresponding to edge $(v,w)$ of $H_i$.
\end{lemma}

When the matrix $A$ is a Hamiltonian matrix used in a QMD simulation, we can assume that $P(A)$ will have a sparsity structure very similar to the structure of the density matrix $D$ from the previous QMD simulation step.
When computing the halos we can then use $G(D)$ instead of $H=\PP(G)$, since the latter graph is not known until $P(A)$ has been computed.
In practice, the current $H$ can also be used to contribute to the halo.

Using $H=G(D)$,
this allows us to construct the following  algorithm for computing $P(A)$:

\begin{enumerate}[(i)]
  \item Divide $V(G)$ into $q$ disjoint sets $\{U_1,\dots,U_q\}$ and define a CH-partition $\Pi=\{\Pi_1,\dots,\Pi_q\}$, where $\Pi_i$ has core $U_i$ and halo $N(U_i,H)\setminus U_i$;
  \item Construct submatrices $A_{U_i}$, $i=1,\dots,q$;
  \item Compute $\PP(A_{U_i})$ for all $i$ independently, using dense matrix algebra;
  \item Define $\PP(A)$ as a matrix whose $i$-th row has nonzero elements equal to the corresponding elements of the $j$-th row of $\PP(A_{U_k})$, where $U_k$ is the set containing vertex $i$ and $j$ is the row in $A_{U_k}$ corresponding to the $i$-th row in $A$. Lemma~\ref{lem:matrixValuesMulti2} shows that such an algorithm computes $P(A)$ accurately.
\end{enumerate}

The computational complexity of computing $P(A)$ by the proposed algorithm is dominated by the complexity of step~(iii),
that is the computation of $\PP(A_{U_i})$ for all $i$, dominated by a dense matrix-matrix multiplication.

If $c_i$ and $h_i$ are the size of the core and the halo of $\Pi_i$, then, by formula \eqref{eq:poly2},
computing $\PP(A_{U_i})$ takes $s(c_i+h_i)^3$ operations,
where $s$ is the number of superpositioned operators (see formula \eqref{eq:poly2}).
Here we exclude those operations needed to threshold some elements of the matrices
(which are quadratic in the worst case, but usually linear in $c_i+h_i$).
Since the parameter $s$ is independent of $\Pi$, a CH-partition that minimizes the computational cost of computing $P(A)$ will minimize $\sum_{i=1}^q (c_i+h_i)^3$. 

Hence, we will be looking at the following
\textit{CH-partitioning problem}. Let $G$ be an undirected graph and  $q\geq 2$ be an integer.
Find a CH-partition of $G$ into $q$ parts $\Pi_1,\dots,\Pi_q$, where $\Pi_i$ has core $U_i$ and halo $N(U_i,G)\setminus U_i$ of sizes $c_i$ and $h_i$, respectively, that minimizes

\begin{equation}
\label{eq:sumOfCubes}
\sum_{i=1}^q (c_i+h_i)^3.
\end{equation}

\section{Proposed Partitioning Algorithms}
\label{section_algorithms}

We investigate several approaches to compute CH-partitions with the aim of minimizing the objective function \eqref{eq:sumOfCubes}
using existing graph partitioning packages as well as our own heuristic algorithm. 
We chose \textit{METIS} and \textit{hMETIS} for our experiments due to their widespread use and \textit{KaHIP} because of its good performance demonstrated at the 10th DIMACS Implementation Challenge~\cite{DBLP:conf/dimacs/2012}.

\subsection{Using Standard Graph Partitioning}
Observe that if we were to take sums without the cubes in the objective function \eqref{eq:sumOfCubes},
we would obtain $|V(G)| + \sum_i h_i$. 
In other words, if we ignore the cubes, we will need to minimize the sum of the halo nodes over all parts. 

One can easily establish a correspondence between regular graph partitions and CH-partitions.
Given a regular partition $P$, we can define a CH-partition $\Pi$ that has cores corresponding to the parts of $P$ and halos defined as the adjacent vertices of the corresponding core vertices,
but not in those cores themselves.
Then, if $(v,w)$ is a cut edge of $P$, one of $v$ or $w$ is a halo vertex. Conversely, if $v$ is a halo vertex of some part in $\Pi$, then there exists a core vertex  $w$ such that $(v,w)$ is a cut edge. So, clearly, the set of the cut edges of $P$ and the set of halo nodes in $\Pi$ are related, but they are clearly also different. However, tools like \textit{METIS} allow  also optimization with respect to the \textit{total communication volume}, which exactly corresponds to the sum of halo nodes. 
Hence, we want to study, by ignoring the cubes in \eqref{eq:sumOfCubes}, how well regular graph partitioning tools can be used to produce CH-partitions.
As an alternative approach, we propose later in this section a heuristic to improve the solution obtained by standard graph partitioning tools.
We consider the following approaches:

\subsubsection{METIS}
\label{subsection_metis}
\textit{METIS} \citep{KarypisKumar1999} is a popular heuristic multilevel algorithm to perform graph partitioning based on a three-phase approach:
(i) The input graph is coarsened by generating a sequence of graphs $G_0, G_1, \ldots, G_n$ starting from the original graph $G=G_0$ and ending with a suitably small graph $G_n$ (typically less than 100 vertices).
(ii) $G_n$ is partitioned using some other algorithm of choice.
(iii) The partition is projected back from $G_n$ to $G_0$ through $G_{n-1},\dots,G_1$.
As each of the finer graphs during the uncoarsening phase contains more degrees of freedom than the multinode graph, 
a refinement algorithm such as Fiduccia-Mattheyses'  \citep{FiducciaMattheyses1982} is used to enhance the partitioning after each projection.
\textit{METIS} has multiple tuning parameters including the size of $G_n$, the coarsening algorithm, and the algorithm used for partitioning $G_n$. 

\subsubsection{KaHIP}
\textit{KaHIP}~\cite{sandersschulz2013} is a family of graph partitioning programs, including several multilevel graph partitioning algorithms.
Like \textit{METIS}, \textit{KaHIP} contracts a given graph, computes partitions on each contraction level
and uses local search methods to improve a partitioning induced by the coarser level.
It offers various heuristics such as local improvements based on max-flow/min-cut
\citep{SandersSchulz2011,FordFulkerson1956},
repeatedly applied Fiduccia-Mattheyses calls \citep{FiducciaMattheyses1982} or F-cycles \citep{SandersSchulz2011}. 

\subsubsection{Using hypergraph partitioning}
\label{subsection_hmetis}
In the hypergraph formulation, the set of all neighbor vertices of each vertex is defined as the single hyperedge corresponding to that vertex.
Using hyperedges has the advantage that, by minimizing edge-cut with respect to hyperedges, either all vertices or no vertex
for a particular set of neighbor nodes are included in a partition,
thus taking care of the halo nodes by itself.

For hypergraph partitioning we use \textit{hMETIS} \cite{KarypisKumar2000}, which is the hypergraph partitioning analog of \textit{METIS}.

\subsection{Simulated Annealing Refinement}
\label{subsection_SA}
Since standard graph and hypergraph partitioning algorithms use an objective function
(the size of the edge or hyperedge cut) which differs from the objective function \eqref{eq:sumOfCubes}, we designed an algorithm that
explicitly minimizes  \eqref{eq:sumOfCubes}.

The simulated annealing (\textit{SA}) optimization approach
\citep{Kirkpatrick1983} is a standard (probabilistic) tool in optimization.
SA is a gradient-free method that
iteratively improves an objective function by proposing a sequence of random 
modifications to an existing solution.
Modifications that minimize the current best solution are always accepted,
while all other moves may be accepted with a certain \textit{acceptance probability}
which depends on two quantities:
first, the acceptance probability of any move is proportional to the magnitude of the (unfavorable) increase in the objective function
resulting from the proposed modification, and second it is antiproportional
to the runtime, meaning that modifications are more likely to be accepted at the start of
each run (the exploration phase of the SA algorithm).
The latter property is implemented with the help of a so-called
\textit{temperature} function which is decreased after each iteration.

Our implementation
(the pseudo-code is given as Algorithm \ref{algorithm_sa} in Appendix \ref{section_details_SA}) 
starts with a CH-partition $\Pi$, which is either randomly 
generated, or produced by another algorithm (for instance by \textit{METIS} or any other partitioning tool).
At the $i$-th iteration, a random edge joining a core vertex $v$ and a halo 
vertex $w$ from the same part  $\Pi_i$ of $\Pi$ is  randomly chosen.
Next, a partition $\Pi'$ is created out of $\Pi$ by moving $w$ from the halo to 
the core of $\pi$.
The proposed modified partition $\Pi'$ is evaluated by computing the change 
$\Delta$ in the objective function \eqref{eq:sumOfCubes} between $\Pi'$ and 
$\Pi$.
The modification $\Pi'$ is accepted with probability $p$, where $p=1$ if 
$\Delta<0$, i.e., $\Pi'$ is better than $\Pi$, and $p<1$ if $\Delta\geq 0$.
Specifically, in the latter case, we set the probability of accepting a modification that increases
the objective function  to
$p=\exp \left(-\Delta/t(i) \right)$.
If the modification is accepted, then we set $\Pi=\Pi'$ and repeat the iteration.

We choose $t(i)=1/i$ as the temperature function.
The stopping criterion of our simulated annealing method was the maximal number of iterations, set to $N=100$.

\section{Summary of Experimental Results}
\label{section_results}

\label{subsection_expresults_goals}
In this section we compare the quality of the CH-partitions computed by the algorithms presented in Section \ref{section_algorithms}
with respect to three measures:
the magnitude of the objective function (\ref{eq:sumOfCubes}),
the time needed to compute the partitioning, and
the scaling behavior of the G-SP2 algorithm for one of these systems
as a function of both the number of MPI ranks and threads.
Instead of using simulated (random) graphs as test systems,
we derive graphs from representations of actual molecules.

\subsection{Choice of Parameters and Experimental Setup}
\label{subsection_parameters}
We employed \textit{METIS} and \textit{hMETIS} with the following parameters obtained through a grid search over sensible values.
The best set of parameters (where ``best'' is understood as the best performance
on average for all systems considered in this section) provided below is kept fixed throughout the section.

\textit{METIS} was run using the default multilevel $k$-way partitioning and the default
sorted heavy-edge matching for coarsening the graph.
The $k$-way partitioning routine of \textit{METIS} allows the user to choose to minimize either with respect to the
\textit{edge-cut} or the \textit{total communication volume} of the partitioning. 
In particular, the definition of the ``total communication volume'' of \textit{METIS}
precisely corresponds to what we defined as the ``sum of halo nodes'' (Section \ref{section_theory}). 
Therefore, it makes sense to choose to minimize with respect to the total communication volume. 

For \textit{hMETIS}, we found the recursive bisectioning routine to perform better than the $k$-way partitioning
when using the following parameters
(see the \textit{hMetis} manual~\cite{hmetis_manual} for the meaning of these parameters):
the vertex grouping scheme \textit{Ctype} was set to $1$ (hybrid first-choice scheme HFC),
the refinement heuristic was set to $1$ (Fiduccia-Mattheyses), and the $V$-cycle refinement scheme
was set to $3$ to perform a $V$-cycle refinement on each intermediate solution.

We tested all sequential implementations on a MacBook Pro laptop running OS X Yosemite. 
Our implementations were written in $C$ and our executables were
compiled for a 64-bit architecture.

\subsection{Test cases motivated by physical systems}
\label{subsection_systems}

We evaluate all algorithms considered in this study (Section \ref{section_algorithms}) on a variety of physical test systems.
These systems are chosen in such a way as to cover a representative set of realistic scenarios where graph partitioning
can be applied to MD simulations.
We give insights into the physics of each test system
that gave rise to the density matrix
and how the structure of the graph affects the results (Section \ref{subsection_assessment_all}).

\begin{figure}
  \begin{center}
   \includegraphics[width=0.27\textwidth]{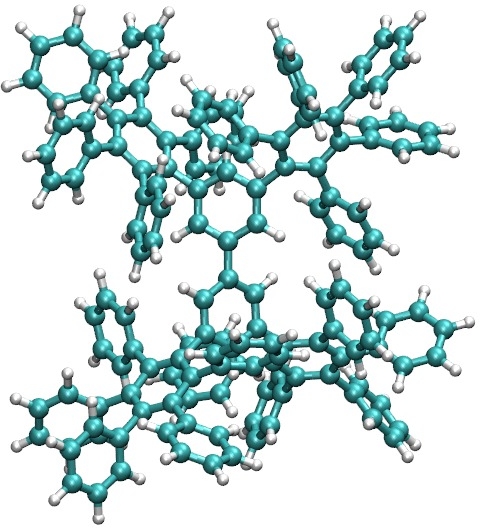}
   \includegraphics[width=0.32\textwidth]{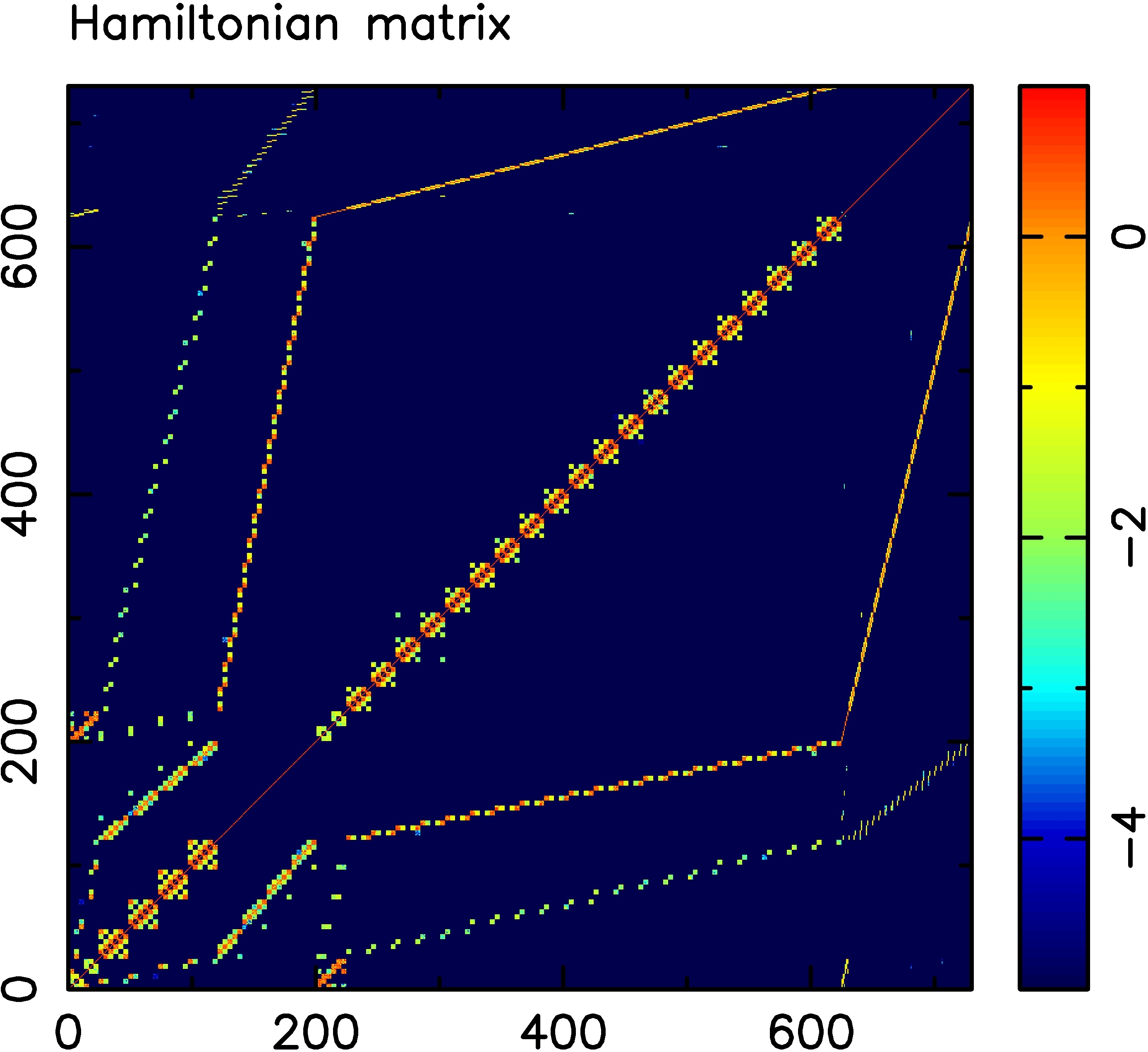}
   \includegraphics[width=0.32\textwidth]{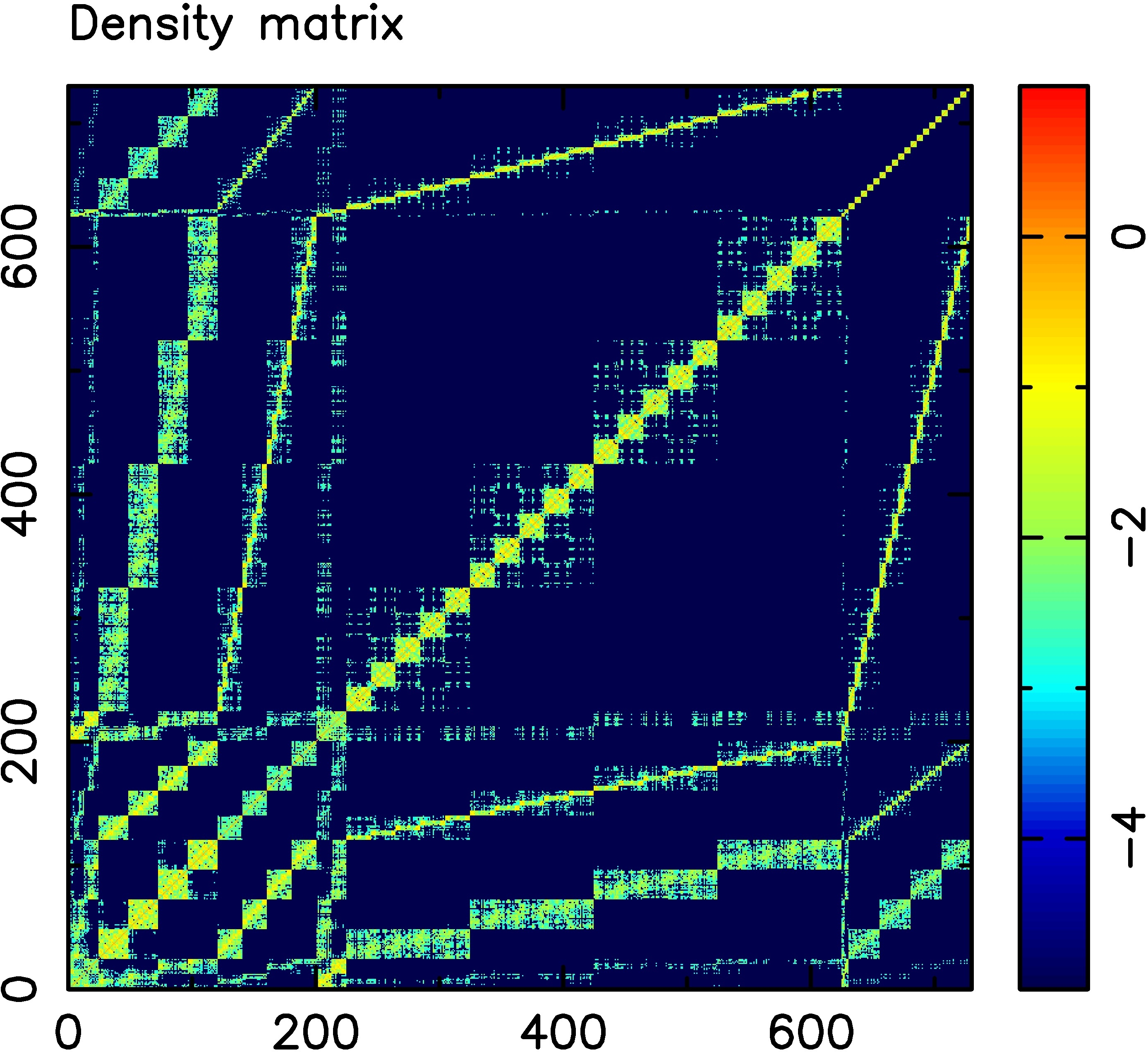}
   \caption{\small Phenyl dendrimer system with its molecular representation (left).
   Carbon and hydrogen atoms are depicted in cyan and white spheres respectively.
   2D plot representation of the Hamiltonian (middle) and thresholded density matrix (right) respectively.
   Here we have plotted the Log$_{10}$ of the absolute values of the elements for each of the matrices.
   The density matrix was computed using the SM-SP2 algorithm.\\[4ex]}
   \label{fig:dendrimer_molecule}
  \end{center} 
\end{figure}

As an example,
Figure \ref{fig:dendrimer_molecule} (left) shows a dendrimer molecule composed of
22 covalently bonded phenyl groups with C and H atoms only.
It has 262 atoms and 730 orbitals, meaning that the resulting graph consists of 730 vertices.

The Hamiltonian matrix (showing absolute values of elements) for this system is shown in Figure \ref{fig:dendrimer_molecule} (middle).
Applying the SM-SP2 algorithm to the Hamiltonian results in the density matrix on the right of Figure \ref{fig:dendrimer_molecule} which encodes its physical properties.

To convert the density matrix into a graph in order to find meaningful physical components via graph partitioning,
the density matrix is thresholded by $10^{-5}$.

We apply the above procedure to all the systems summarized in Table \ref{tab:systems}
in order to obtain adjacency matrices.
\begin{table}[t]
\scriptsize
\caption{Different physical systems used to evaluate different methods.
$n$ is the number of vertices in the graph and $m$ is the number of edges.\label{tab:systems}}
  \begin{center}
    \begin{tabular}{|l|l|r|r|r|l|}
    \hline
      No & Name & $n$ & $m$ & $m/n$ & Description \\ \hline
      1 & polyethylene dense crystal	& 18432	& 4112189	& 223.1	& crystal molecule in water solvent with very low threshold\\
      2 & polyethylene sparse crystal	& 18432	& 812343	& 44.1	& crystal molecule in water solvent with very high threshold\\
      3 & phenyl dendrimer 		& 730 	& 31147 	& 42.7	& polyphenylene branched molecule\\ 
      4 & polyalanine 189 		& 31941	& 1879751	& 58.9	& poly-alanine protein solvated in water\\
      5 & peptide 1aft 			& 385 	& 1833 		& 4.76	& ribonucleoside-diphosphate reductase protein \\
      6 & polyethylene chain 1024 	& 12288	& 290816 	& 23.7	& chain of polymer molecule, almost 1-d\\ 
      7 & polyalanine 289 		& 41185	& 1827256	& 44.4	& large protein in water solvent\\
      8 & peptide trp cage 		& 16863	& 176300	& 10.5	& smallest protein with ability to fold,\\
      &&&&&surrounded by water molecules\\
      9 & urea crystal			& 3584	& 109067	& 30.4	& organic compound involved in many processes\\
      &&&&&in living organisms  \\ \hline
  \end{tabular}
\end{center}
\end{table}
The table shows the name of the molecule system in the first column, together
with its number of vertices $n$ and edges $m$.

\subsection{Evaluation of the partitioning algorithms on a variety of real systems}
\label{subsection_assessment_all}
\begin{figure}[t]
  \centering
  \includegraphics[width=0.7\textwidth]{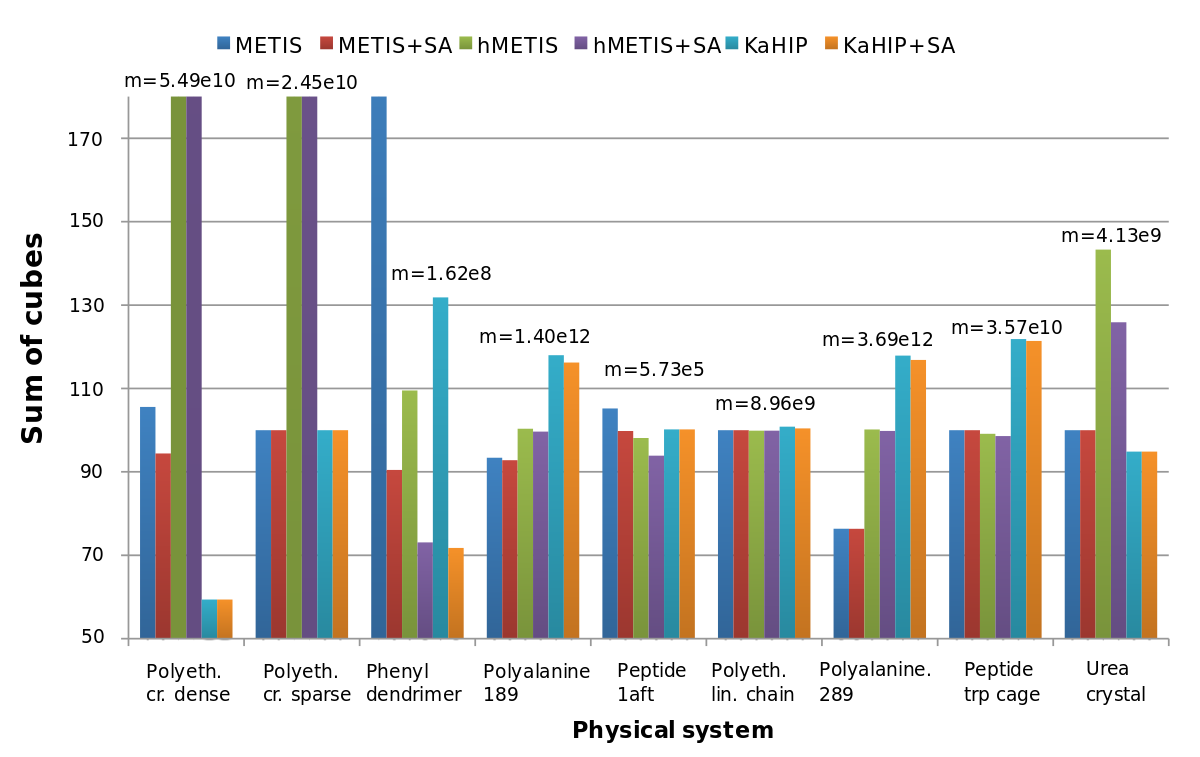}
  \caption{Performance of different methods on our test systems measured using the sum of cubes criterion.
  For each experiment, the values have been normalized by setting each median to 100.
  Bars representing very large values have been truncated to make the rest of the chart more informative.
  (The exact values are available in Appendix~\ref{app:experiments}, Table~\ref{tab:results}.) 
    \label{fig:charts} }
\end{figure}

\begin{figure}[t]
  \centering
  \includegraphics[width=0.7\textwidth]{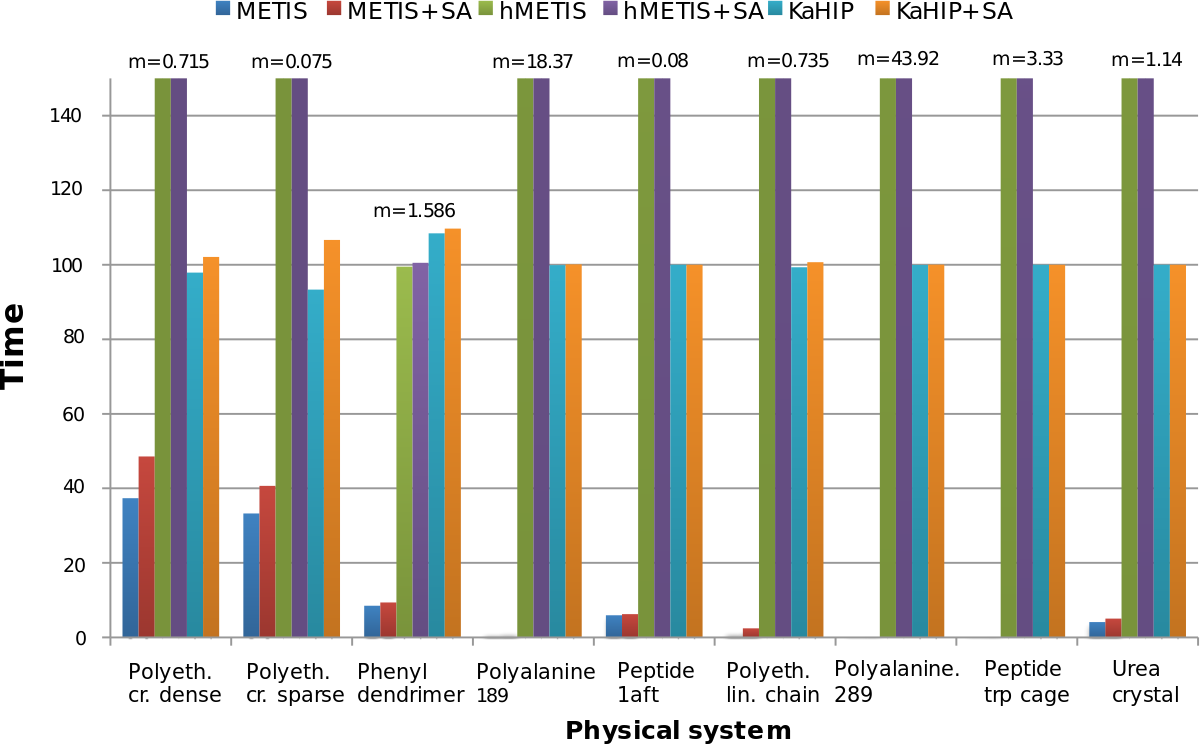}
  \caption{Performance of different methods on our test systems measured in terms of computing time for the partitioning.
  To handle the big disparity between values for different graphs, the same formatting as in Figure~\ref{fig:charts} have been used.
  \label{fig:charts2}}
\end{figure}

We partition each graph into 16 parts using six methods:
(i) \textit{METIS}  used with the parameters given in Section \ref{subsection_parameters};
(ii) \textit{METIS} followed by simulated annealing (SA);
(iii) \textit{hMETIS}; (iv) \textit{hMETIS} in combination with SA; (v)
\textit{KaHIP}; (vi) \textit{KaHIP} plus SA.
The effectiveness of the methods  are  evaluated using the sum of cubes \eqref{eq:sumOfCubes} criterion.

Results of our experiments are summarized in Figures~\ref{fig:charts} and~\ref{fig:charts2}. 
The immediate observation is, except for the first two systems, all algorithms perform  well,
even though \textit{METIS} and \textit{KaHIP} are considerably faster than \textit{hMETIS}.
The usage of SA as a post-processing step seems favorable as it is able to improve
the solution returned by existing methods considerably in almost all cases at negligible additional runtime.

\textit{hMETIS} seems to be somewhat unsuited for this flavor of the partitioning problem
as its solutions are usually worse in quality than those of the other two methods.
Also, its runtime greatly exceeds other methods, making it unsuited for QMD
simulations over longer time intervals, the ultimate goal of this work.
The explanation of the exceptionally bad behavior for the first two systems needs to be investigated further.

\textit{KaHIP} yields very good solutions in terms of the sum of cubes.
Nevertheless, it is outperformed by the usage of \textit{METIS} and a SA post-processing step,
which often yields considerably better results and, moreover, has a shorter combined computational runtime.

The sparsity of the graph for a physical system seems to be of importance for the behavior of the algorithms.
Whereas for denser systems \textit{METIS} outperforms \textit{hMETIS}, this is not the case for sparser ones.
Moreover, SA seems to be able to further improve solutions especially well in the more dense cases.
This can easily be explained as dense cases offer more possibilities to move and optimize edges after partitioning than sparse ones.
For example, the dendrimer is very dense considering its small number of vertices
and hence the combination \textit{METIS+SA} performs very well
(details given in Appendix~\ref{app:experiments}, Table~\ref{tab:results}).

\subsection{Parallel G-SP2}
\label{subsection_scaling_protein}
Instead of relying on the sum of cubes criterion (\ref{eq:sumOfCubes}) we also measured
the quality of the CH-partitions through the speed-up obtained when parallelizing the G-SP2 algorithm
and applying it to real physical systems.

\begin{figure}
  \includegraphics[width=0.49\textwidth]{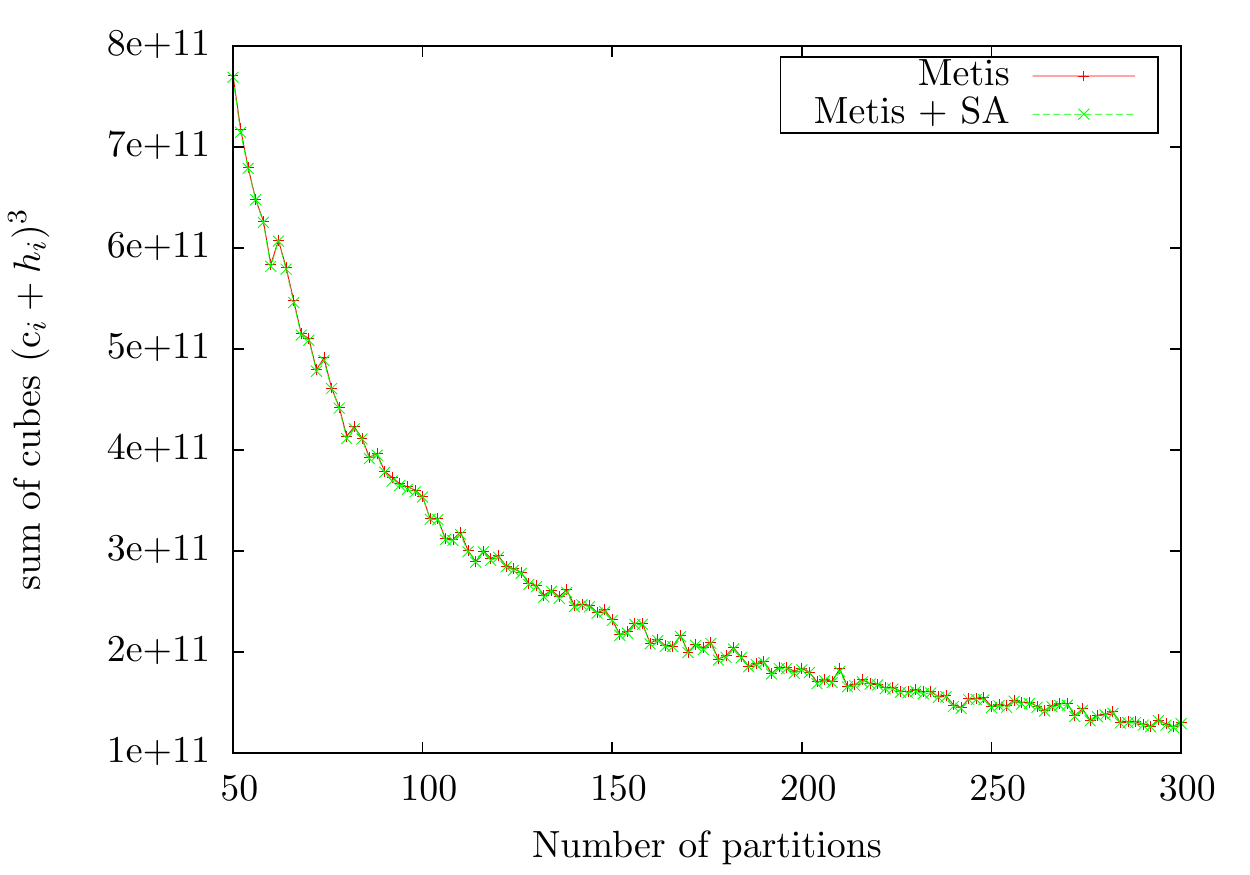} \hfill
  \includegraphics[width=0.49\textwidth]{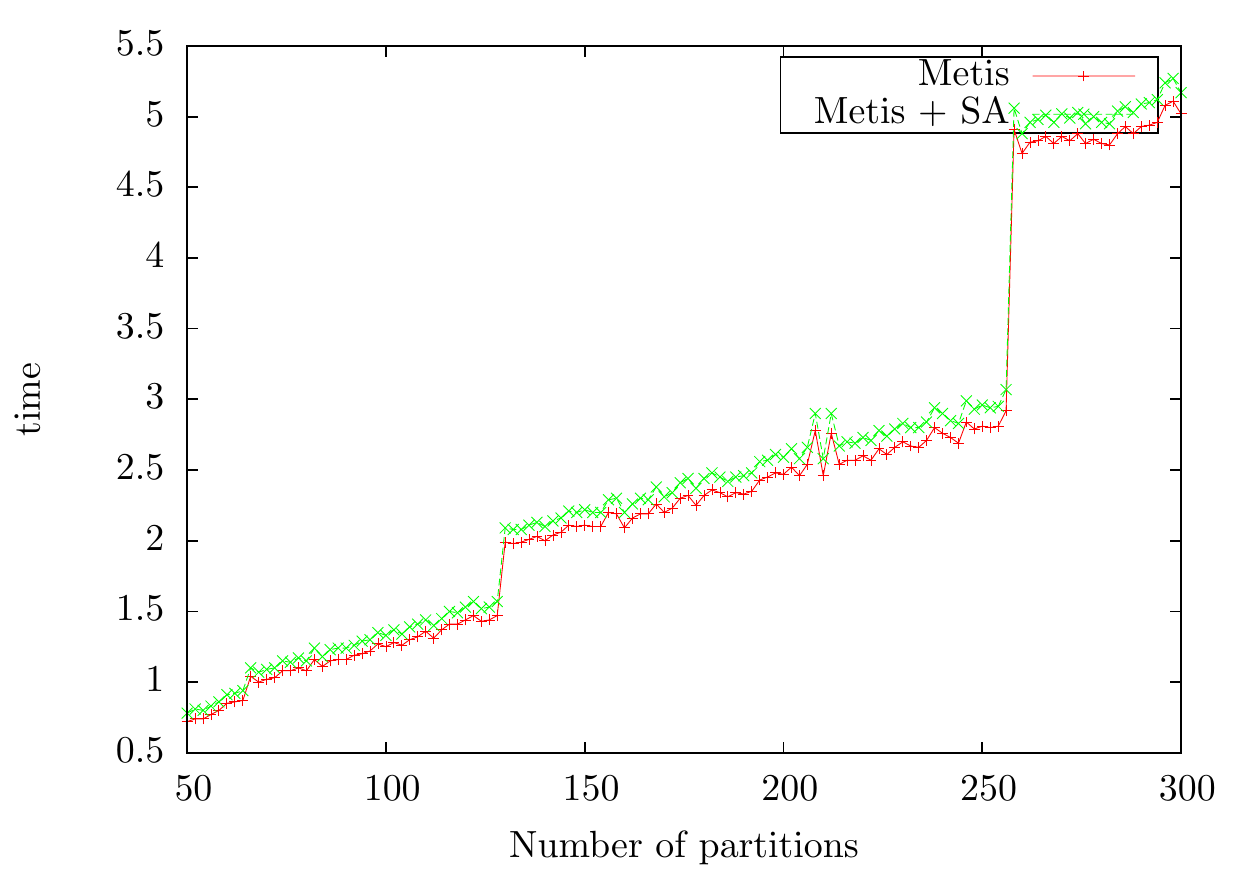} 
  \caption{Dependence of the number of parts on the  sum of cubes measure
  and the time to find the partition.\label{fig:protein_scale}}
\end{figure}

We used CH-partitions obtained by the \textit{METIS} and \textit{METIS+SA} methods on a large protein system
(labeled ``polyalanine 259'' in Table \ref{tab:systems} with $n=41,185$ vertices and $m=21,827,256$ edges).
The experiments were carried out on the Wolf IC cluster of Los Alamos National Laboratory. 
Each computing node has 2 sockets housing an 8-core Intel Xeon SandyBridge E5-2670 with a total of 16 cores per computing node. 
Parallelism was achieved by OpenMPI (for parallelism across nodes) and OpenMP (for parallelism across cores within a node).

Figure~\ref{fig:protein_scale} shows how the sum of cubes measure
and the computing time varies as we increase the number of CH-partitions for the large protein system.
According to the left plot, the total effort of the G-SP2 algorithm,
measured in terms of the sum of cubes criterion, decreases steadily as the number of partitions and parallelized subproblems
increases.

Corresponding to the left plot of Figure~\ref{fig:protein_scale},
the right plot shows the effort for the graph partitioning step alone, measured in terms of computing time.
As expected, the total effort increases with an increasing number of partitions.
Noticeable are the steps in the plot occurring at $65$, $129$, $257$ etc.\ partitions. These are due
to the fact that the multilevel implementation of \textit{METIS} bisects the partitioning problem
into one more (recursive) layer each time the number of partitions surpasses a power of two.
Moreover, Figure~\ref{fig:protein_scale} (right plot) demonstrates that the additional effort added by employing
the SA post-processing step is minimal in comparison to the actual graph partitioning.
Given the improvements achieved by post-applying SA to the edge-cut optimized partitions computed by a conventional algorithm,
its usage seems very sensible.
Overall, Figure~\ref{fig:protein_scale} visualizes that despite the increasing effort to compute a partitioning,
the total effort of the G-SP2 algorithm decreases in a parallelized application.

\subsection{Comparison of SM-SP2 on a single node to parallel G-SP2}
\label{subsection_scaling_prot}
\begin{figure}
\centering
\includegraphics[width=0.7\textwidth]{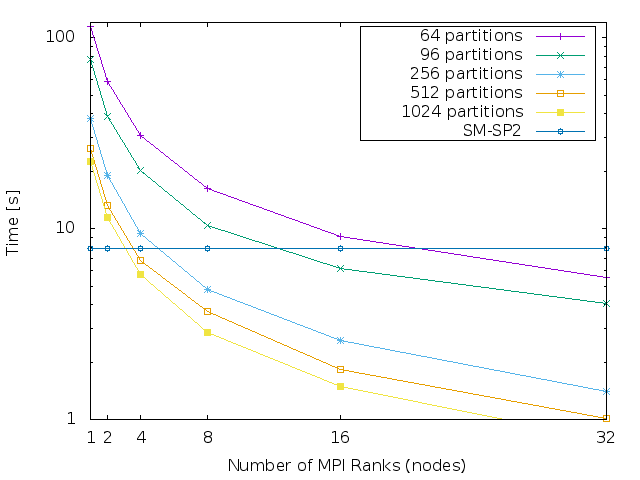}
\caption{Runtime against number of nodes for a parallelized G-SP2 using different
numbers of partitions for the polyalanine 259 molecule.
\label{fig:SueChristian}}
\end{figure}
Figure \ref{fig:SueChristian} shows the runtime for a parallel G-SP2 run across 1-32 nodes and
compares the runtime against a threaded single node implementation of SM-SP2 for the polyalanine 259.
We use a single node for SM-SP2 since, in a multi-node implementation, the communication overhead exceeds the gain offered by the extra computing power; which issue is the main motivation for developing G-SP2.
As before, we employed \textit{METIS} with parameters specified in Section \ref{subsection_parameters}
together with a SA post-processing step.

As visible from Figure \ref{fig:SueChristian} the G-SP2 runtime decreases both with the numbers of nodes
as well as with the numbers of partitions used.
This is as expected.
The decrease is most prominent when only few nodes are used for parallelization since then,
increasing the number of parallel nodes causes the runtime to drop sharply.
The curves somewhat flatten out for higher number of nodes.

Interestingly, for low numbers of nodes (up to between $4$ and $16$ nodes depending on the number of partitions used)
the overhead from the parallel G-SP2 computation causes the
partitioned run to be slower than the SM-SP2 computation on a single node.
As more nodes are used, the runtime decreases and falls below the one of a single node implementation:
Figure \ref{fig:SueChristian} shows that for this particular physical system
at least $4$ nodes need to be used to observe a speed-up in computation.

\section{Partitioning-system relationship}
\label{subsection_partitioning_relationship}
\begin{figure}[t]
  \begin{center}
   \includegraphics[width=0.7\textwidth]{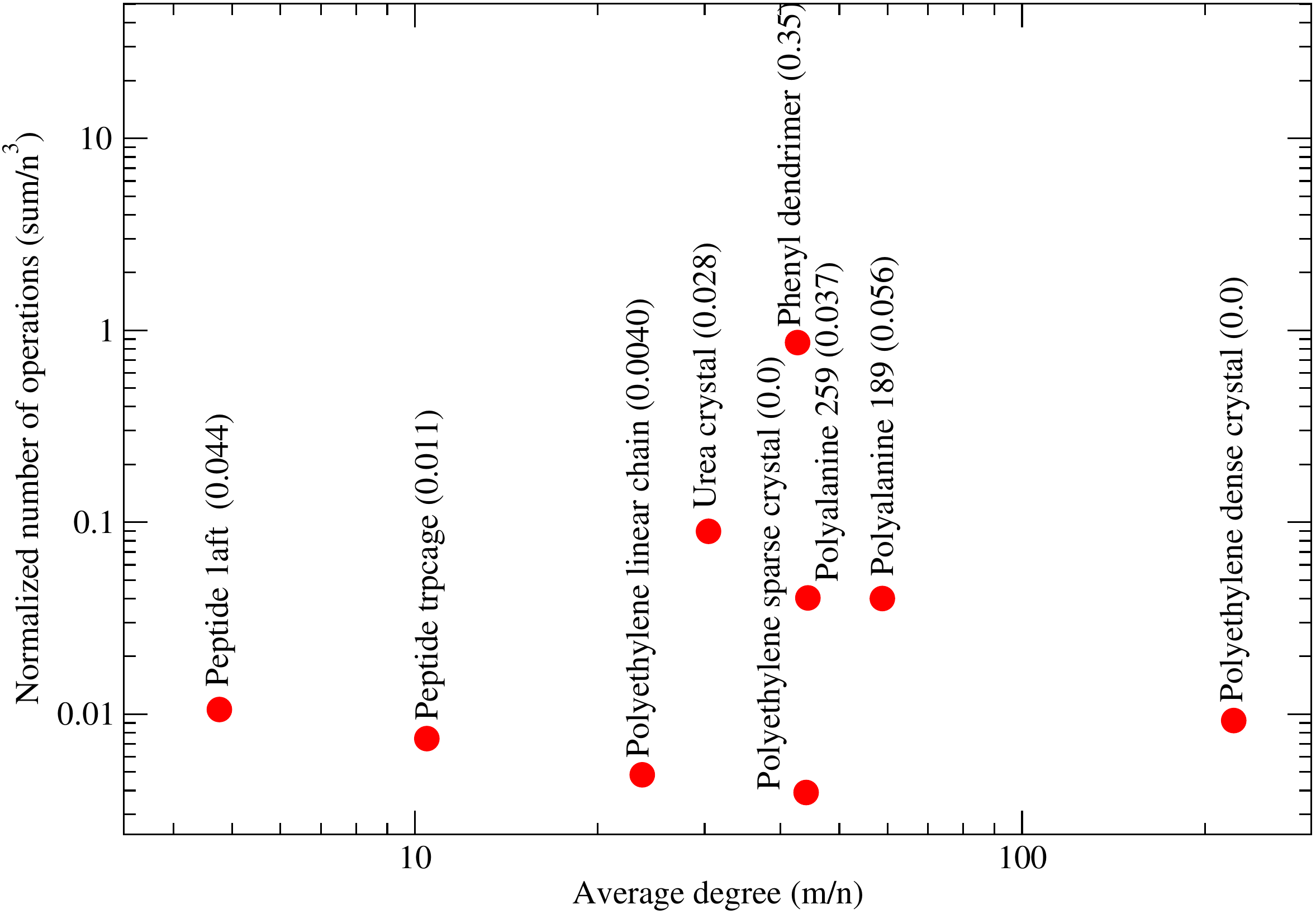}   
   \caption{Normed number of operations, defined as the value of \eqref{eq:sumOfCubes} divided by $n^3$, the complexity of dense matrix-matrix multiplication, where $n$ is the number of vertices
   found with \textit{METIS} and average degree $m/n$ for all the systems in Table \ref{tab:systems}.
   Fractions of $(\max-\min)/n$ are shown in brackets.
   Similar trends would be expected for the \textit{METIS} + SA algorithm.}
   \label{matrices}
  \end{center}
\end{figure}

It is interesting to note that, in general, there is no correlation between the average degree (connectivity)
and the normalized number of operations (NNO) obtained as a result of the partitioning (Figure~\ref{matrices}).
An example of the latter is observed for the polyethylene dense crystal case where the average degree is large and
the normalized number of operations stays low, similar to other molecules with smaller average degree.

For 1D systems such as the polyethylene linear chain and the polyethylene sparse crystal
(regular agglomerate of polyethylene chains aligned along a particular direction with a large chain-to-chain distance
\citep{TF9393500482}), the algorithm finds the lowest NNO.
The latter is probably due to the high order combined with the sparsity of the system.

Systems that can be viewed as regular 
(polyethylene linear chain, polyethylene sparse crystal, polyethylene dense crystal and urea crystal)
do not seem to have any advantage over the rest of the systems in terms of NNO obtained values.

Phenyl dendrimer is a special case where the molecule has an assotiated graph with a fractal-like structure
and as shown by the NNO values it represents a difficult case for graph partitioning.
The combination of \textit{METIS} and SA seems to yield a large improvement
as can be seen from Table \ref{tab:results}.

The case of proteins (solvated polyalanines) seems to have large NNO which could be
attributed to the large average of node degree when compared to peptides (small proteins). 
Furthermore,
there seems to be a correlation between the maximum and minimum partition norm difference (MMPN)
computed as $(\max-\min)/n$ and the NNO.
For the case of the dendrimer this is clear as it shows both large MMPN and NNO values.
The case of proteins are intermediate cases with both intermediate values of MMPN and NNO and finally,
sparse ordered systems such as polyethyene chains have both low values of MMPN and NNO.

\section{Conclusion}
\label{section_discussion}
This paper considers graph partitioning as a means to parallelize (and hence to speed up) the computation of the density matrix
in MD simulations.
The graphs to be partitioned are obtained from representations of the density matrices underlying physical systems.

The main focus of this article is to study
a version of the graph partitioning problem
in which partitions are minimized with respect to both
the parts' core sizes as well as the number of their neighbors in adjacent partitions (halos),
which has not been studied previously to the best of our knowledge.
We called this flavor of graph partitioning the \textit{CH-partitioning} (core-halo partitioning).

The contributions of this work are twofold.
First, it provides a mathematical description and justification of the CH-partitioning problem under consideration.
Importantly, we derive a sufficient
condition ensuring that a partitioned evaluation of a matrix polynomial yields the same result
as the evaluation of the original (unpartitioned) matrix.
Second, we investigate several approaches to compute CH-partitions and evaluate
them using three error criteria, namely the total computational effort, the maximal effort per processor,
and the overall computational runtime.
We pay special attention to a modified SA approach which is used to post-optimize partitions
computed by conventional graph partitioning algorithms.

Our experiments show that standard graph partitioning packages can be used to solve our flavor of the partitioning problem,
and that SA proves to be well suited to post-optimize partitions obtained through them.
As a recommendation, we conclude that \textit{METIS} with a SA post-processing step
is the best approach as it is significantly faster than competing methods
while giving the best results on average.

\subsubsection*{Acknowledgments}
The authors would like to thank Purnima Ghale and Matthew Kroonblawd
for their help with selecting meaningful physical datasets of real-world molecules and
Ben Bergen, 
Nick Bock, 
Marc Cawkwell, 
Christoph Junghans, 
Robert Pavel, 
Sergio Pino, 
Jerry Shi,
and Ping Yang for their feedback.

We acknowledge support from Office of Basic Energy Sciences
and the Laboratory Directed Research and Development program of Los Alamos National Laboratory.
This research has been supported at Los Alamos National Laboratory under the Department of Energy contract DE-AC52-06NA25396.

\appendix
\section{Proofs for Section \ref{section_theory}}
\label{appendix_proofs}

Let, for any symmetric matrix $X=\{x_{ij}\}$, $G(X)$ denote the graph, called \textit{sparsity graph} of $X$, that encodes the zero-nonzero structure of $X$. Specifically, for the $i$-th row (column) of $X$ there is a vertex $i$ in $G(X)$, and there is an edge between vertices $i$ and $j$ if and only if $x_{ij}\neq 0$.  
For any graph $G$, we denote by $V(G)$ and $E(G)$ the sets of the vertices and edges of $G$, respectively.

Let $A$ be a symmetric $n \times n$ matrix.
We first define a generalization of a matrix polynomial in formula \eqref{eq:expansion}.
We define a \textit{thresholded matrix polynomial} of degree $m=2^s$ to be a superposition of operators of the type
\begin{equation}\label{eq:poly}
P=P_1 \circ T_1 \circ \dots \circ P_s \circ T_s,
\end{equation}
where $P_i$ is a polynomial of degree 2 and $T_i$ is a thresholding operation. Formally, $T_i$ is a graph operator associated with a set of edges $E(T_i)$ such that, for any graph $I$, $T_i(I)$ is a graph with a vertex set $V(I)$ and an edge set $E(I)\setminus T_i$. 

Denote by $P(A)$ the application of a superpositioned operator $P$ as defined above,
consisting of polynomials $P_i$ and thresholding operations $T_i$,
to a matrix $A$ of appropriate dimension.
In the motivating SM-SP2 application, $A$ corresponds to the Hamiltonian and $P(A)$ corresponds to the density matrix.

Define the structure class  $\M(A)$ of $A$ as the set of all matrices $B$ such that $G(A)=G(B)$,
meaning that all matrices in this class have the same zero-nonzero structure as $A$.
Let $P$ be a thresholded matrix polynomial and $G=G(A)$. Define $\PP(G)$ as the minimal graph with the same vertices as $G$ such that, for any matrix $B\in \M(A)$ and any $v,w$ such that $P(B)|_{vw}\neq 0$, there is an edge $(v,w)\in E(P(G))$. 

Informally, $\PP(G)$ describes the worst case zero-nonzero structure of $P(A)$ that ignores the possibility of coincidental zeros resulting from cancellation (adding opposite-sign numbers). We assume that all diagonal elements of $A$ are non-zero and that no $E(T_i)$ contains a loop edge; therefore, there is a loop edge associated with each vertex of $G$ and $\PP(G)$.

For any graph $I$ and vertex $v$ of $I$, the \textit{neighborhood} of $v$ in $I$ is the set $N(v,I)=\{w\in V(I)~|~(v,w)\in E(I)\}$.
Let $H=\PP(G)$, $v$ be a vertex of $G$, and $H_{v}$ denote the subgraph of $H$ induced by $N(v,H)$.
For the following lemmas we assume that $T_i\cap E(H)=\emptyset$ for all $i$, i.e., none of the edges in $H=\PP(G)$ are thresholded.
We have the following properties.

\begin{lemma}\label{lem:neighb}
Let $v$ be a vertex of $G$. Then $N(v,\PP(G))=N(v,\PP(H_v))$.
\end{lemma}

\begin{proof}
First we prove that $N(v,\PP(G))\subseteq N(v,\PP(H_v))$. Let $w\in N(v,\PP(G))$. Then $(v,w)\in E(\PP(G))$ and hence $(v,w)\in E(H_v)\subseteq E(H)$. Since by assumption $T_i\cap E(H)=\emptyset$, $(v,w)\not\in T_i$ for all $i$. From $(v,w)\in E(H_v)$, the last relation, and the fact that all vertices of $H$ have loops, $(v,w)\in E(\PP(H_v))$. Hence, $w\in N(v,\PP(H_v))$.

Now we prove that $N(v,\PP(H_v))\subseteq N(v,\PP(G))$. Let $w\in N(v,\PP(H_v))$. Since $\PP(H_v)$ and $H_v$ have the same vertex sets, then $w\in N(v,H_v)$. Furthermore, since $H_v$ is a subgraph of $H$, $w\in N(v,H)=N(v,\PP(G))$.

\end{proof}

The lemma shows that $v$ has the same neighbors in $\PP(G)$ and $\PP(H_v)$, i.e., their corresponding matrices have nonzero entries in the same positions in the row (or column) corresponding to $v$. We will next strengthen that claim by showing that the corresponding nonzero entries contain equal values.

Let $X_{v}$ be the submatrix of $A$ defined by all rows and columns that correspond to vertices of $V(H_{v})$. We will call vertex $v$ the \emph{core} and the remaining vertices \textit{halo} of $V(H_{v})$. We define the set $\{V(H_{v})~|~v\in G\}$ a \emph{CH-partition} (from core-halo) of $G$. Note that, unlike other definitions of a partition used elsewhere, the vertex sets of CH-partitions (and, specifically, the halos) can be, and typically are, overlapping.

\begin{lemma}\label{lem:matrixValues}
For any $v\in V(G)$ and any $w\in N(v,\PP(G))$, the element of $P(A)$ corresponding to edge $(v,w)$ of $\PP(G)$ is equal to the element of $P(X_v)$ corresponding to edge $(v,w)$ of $\PP(H_v)$.
\end{lemma}

\begin{proof}
Let $m=2^s$ be the degree of $P$. We will prove the lemma by induction on $s$. Clearly, the claim is true for $s=0$ since the elements of both $A^1$ and $X^1$ are original elements of the matrix $A$. Assume the claim is true for $s-1$.
Define $P'=P_1 \circ T_1 \circ \dots \circ P_{s-1} \circ T_{s-1}$. By the inductive
assumption, the corresponding elements in he matrices $A'=P'(A)$ and $X'=P'(X)$ have equal values. We need to prove the same for the elements of $A'^2$ and $X'^{2}$.

Let $(v,w)\in E(\PP(G))$. By Lemma~\ref{lem:neighb}, $(v,w)\in E(\PP(H_v))$. For each vertex $u$ of $\PP(H_v)$ let $u'$ denote the corresponding row/column of $X$. We want to show that $P(A)(v,w)=P(X)(v',w')$.

By definition of matrix product, $A'^2(v,w)=\sum A'(v,u) A'(u,w)$, where the summation is over all $u$ such that $(v,u),(u,w)\in E(\PP(G))$. Similarly, $X'^2(v',w')=\sum X'(v',u')X'(u',w')$, where the summation is over the values of $u'$ corresponding to the values of $u$ from the previous formula, by Lemma~\ref{lem:neighb}. By the inductive assumption, $A'(v,u)=X'(v',u')$ and $A'(u,w)=X'(u',w')$, whence $A'^2(v,w)=X'^2(v',w')$.
Since by assumption $A'(v,w)=X'(v',w')$, then $P_s(A')(v,w)=P_s(X')(v',w')$, and hence
$P(A)(v,w)=(P_s \circ T_s)(A')(v,w)=(P_s \circ T_s)(X')(v',w')=P(X)(v',w')$.
\end{proof}

Lemma~\ref{lem:matrixValues} implies the following algorithm to compute $\PP(A)$ given we know its sparsity structure in $H_v$:

\begin{enumerate}[(i)]
  \item Construct a CH-partition $\Pi$ of $G$ into $n$ parts such that each part consists of a vertex (core) and its adjacent vertices in $H_v$ (halo);
  \item For the $i$-th part $\Pi_i$ of $\Pi$ that has core the $i$-th vertex of $G(A)$, construct a submatrix $A_i$ containing the rows and columns of $A$ corresponding to the vertices of $\Pi_i$;
  \item Compute $\PP(A_i)$ for all $i$;
  \item Define $\PP(A)$ as a matrix whose $i$-th row has nonzero elements corresponding to the $i$-th row of $\PP(A_i)$ (subject to appropriate reordering).
\end{enumerate}

Clearly, in many cases it will be advantageous to consider CH-partitions whose cores contain multiple vertices. We will next show that the above approach for CH-partitions with single-node cores can be generalized to the multi-node core case.

We will generalize the definitions of $N(v,\PP(G))$ and $N(v,\PP(H_v))$ for the case where vertex $v$ is replaced by a set $U$ of vertices of $G$. For any graph $I$, we define $N(U,I)=\bigcup_{v\in U}N(v,I)$. Furthermore, we define by $H_U$ the subgraph of $H$ induced by $N(U,H)$.

Suppose the sets $\{U_i~|~i=1,\dots,q\}$ are such that $\bigcup_{i}U_i=V(G(A))$ and $U_i\cap U_j=\emptyset$.
In this case we can define a CH-partition of $G=G(A)$ consisting of $q$ sets, where for each $i$,
$U_i$ is the core and $N(U_i,H)\setminus U_i$ is the halo of $\Pi_i$.

The following generalizations of Lemma~\ref{lem:neighb} and Lemma~\ref{lem:matrixValues} follow in a straightforward manner.

\begin{lemma}\label{lem:neighbSet}
Denote by $H_i$ the subgraph $\PP(H_{U_i})$ of $G$.
Let $v$ be a vertex of $U_i$. Then $N(v,\PP(G))=N(v,H_i)$.
\end{lemma}

Denote by $A_{U_i}$ the submatrix of $A$ consisting of all rows and columns that correspond to vertices of $V(H_{U_i})$.
The following main result of this section shows that $P(A)$ can be computed on submatrices of
the Hamiltonian
and hence justifies the parallelized evaluation of a matrix polynomial.

\begin{lemma}\label{lem:matrixValuesMulti}
For any $v\in U_i$ and any $w\in N(v,\PP(G))$, the element of $P(A)$ corresponding to edge $(v,w)$ of $\PP(G)$ is equal to the element of $P(A_{U_i})$ corresponding to edge $(v,w)$ of $H_i$.
\end{lemma}

\section{Pseudocode of the Simulated Annealing Algorithm}
\label{section_details_SA}

\begin{algorithm}[Simulated Annealing]
\label{algorithm_sa}~
\begin{enumerate}
  \setlength{\parskip}{0pt}
  \item Input: Graph G, number of iterations $N$, initial partitioning $\Pi$
  \item Output: Updated partitioning $\Pi$
  \item Select a temperature function $t(i)=1/i$;
  \item For $i=1$ to $N$:
  \begin{enumerate}
    \item Select random partition $\pi$ in $\Pi$ and a random  core-halo edge $(v,w)$ in $\pi$;
    \item In a copy $\Pi'$ of $\Pi$, make $w$ a core vertex of $\pi$ and update the halo of $\pi$;
    \item Compute the values $S$ and $S'$ of \eqref{eq:sumOfCubes} for $\Pi$ and $\Pi'$, respectively, and set $\Delta=S'-S$;
    \item Compute $p = \exp \left(-\Delta/t(i)  \right)$;
    \item Set $\Pi=\Pi'$ with probability $\min(1,p)$;
    \item If a stopping criterion is satisfied exit loop;
  \end{enumerate}
  \setlength{\parskip}{0pt}
  \item Output partitioning $\Pi$;
\end{enumerate}
\end{algorithm}

\section{Experimental Study Results}
\label{app:experiments}
Table~\ref{tab:results} provides the raw data of the various experiments described in Section~\ref{subsection_assessment_all}.
Figure~\ref{fig:maxChart} shows how different methods compare with respect to the maximum size of CH-partitions obtained for the various test systems.

The table visualizes the performance of the six partitioning schemes given in column ``methods'',
measured in four different ways:
Column $3$ (``sum'') shows the sum of cubes criterion (\ref{eq:sumOfCubes}) which is a measure of the total
matrix multiplication cost of a step of the SP2 algorithm.
Since the matrix multiplication consumes most of the computation time in SP2, this is measure of the computational effort of SP2.
The smallest and largest size of any CH-partition given in columns $4$ (``min'') and $5$ (``max'')
is a measure of the spread of partition sizes created by the algorithm: ideally, all partitions should be of
roughly equal size. If this is not the case, the nodes or processors in a parallel implementation of the SP2 algorithm
will have very unequal computational loads which is undesirable in practice.
The last column (``time'') shows the average computation time for each partitioning algorithm measured in seconds.

\begin{table*}
\scriptsize
\caption{Different partition schemes applied to various test systems. 
  The first column is the name of the test system, 
  the second column is algorithmic method used, and the next columns are the corresponding result values 
  (sum of cubes (\ref{eq:sumOfCubes}), the sum of size and halo for the smallest CH-partition (min), biggest CH-partition (max), and
   the overall time to compute the CH-partitioning respectively).
For each test system, the number of vertices $n$, the number of edges $m$ and the number of partitions $p$ are displayed.
\label{tab:results}}
\begin{center}
\begin{tabular}{|l||l|r|r|r|r|}
Test system & method & sum & min & max & time [s]\\
\hline
polyethylene dense crystal &METIS & 57,982,058,496 & 1536 & 1536 &  0.267175 \\ 
n = 18432 & METIS + SA & 51,856,752,364 & 976 & 1536 & 0.347209 \\ 
m = 4112189 & HMETIS & 7,126,357,377,024 & 3840 & 9984 & 141.426 \\ 
p = 16 & HMETIS + SA & 1,362,943,612,944 & 2520 & 5814 & 141.79 \\ 
& KaHIP & 32,614,907,904    &     768    &    1536   &  0.7 \\ 
& KaHIP + SA &  32,614,907,904  &  768   & 1536  &  0.73 \\
\hline

polyethylene sparse crystal &METIS & 24,461,180,928 & 1152 & 1152 &  0.024942 \\ 
n = 18432 & METIS + SA & 24,461,180,928 & 1152 & 1152 & 0.030508 \\ 
m = 812343 & HMETIS & 195,689,447,424 & 2304 & 2304 & 55.9726 \\ 
p = 16 & HMETIS + SA & 170,056,587,295 & 2013 & 2299 & 55.9943 \\ 
& KaHIP & 24,461,180,928     &   1152    &    1152 &   0.07\\ 
& KaHIP + SA &  24,461,180,928   &     1152    &    1152 &   0.08 \\
\hline

phenyl dendrimer &METIS & 336,049,081 & 150 & 409 &  0.13482 \\ 
n = 730 & METIS + SA & 146,550,740 & 0 & 382 & 0.14877 \\ 
m = 31147 & HMETIS & 177,436,462 & 135 & 358 & 1.578 \\ 
p = 16 & HMETIS + SA & 118,409,940 & 0 & 358 & 1.59436 \\ 
& KaHIP & 231,550,645        &  55      &   381  &  1.72\\ 
& KaHIP + SA &  116,248,715    &       0      &   324 &   1.74 \\
\hline

polyalanine 189 &METIS & 1,305,573,505,507 & 3358 & 5145 &  0.332091 \\ 
n = 31941 & METIS + SA & 1,297,206,329,828 & 3362 & 5093 & 0.372463 \\ 
m = 1879751 & HMETIS & 1,402,737,703,273 & 3762 & 5124 & 418.229 \\ 
p = 16 & HMETIS + SA & 1,393,115,476,879 & 3765 & 5119 & 418.28 \\ 
& KaHIP &  1,649,301,823,304       &   12 &       6030  & 18.35 \\ 
& KaHIP + SA & 1,624,800,725,049     &     12    &    5983 &  18.39  \\
\hline

peptide 1aft &METIS & 603,251 & 24 & 41 &  0.004755 \\ 
n = 384 & METIS + SA & 572,281 & 24 & 41 & 0.005007 \\ 
m = 1833 & HMETIS & 562,601 & 24 & 40 & 0.820561 \\ 
p = 16 & HMETIS + SA & 538,345 & 24 & 42 & 0.820771 \\ 
&  KaHIP & 575,978      &    11      &    44  &  0.08  \\ 
& KaHIP + SA &  575,978      &    11       &   44  &  0.08 \\
\hline

polyethylene chain 1024 &METIS & 8,961,763,376 & 800 & 848 &  0.01513 \\ 
n = 12288 & METIS + SA & 8,961,763,376 & 800 & 848 & 0.017951 \\ 
m = 290816 & HMETIS & 8,951,619,584 & 824 & 824 & 27.3297 \\ 
p = 16 & HMETIS + SA & 8,951,619,584 & 824 & 824 & 27.3332 \\ 
& KaHIP & 9,037,266,968     &    782     &    875    & 0.73 \\ 
& KaHIP + SA & 9,000,224,048     &    782      &   872  &  0.74 \\
\hline

polyalanine 289 &METIS & 2,816,765,783,803 & 4591 & 6102 &  0.366308 \\ 
n = 41185 & METIS + SA & 2,816,141,689,603 & 4591 & 6093 & 0.399265 \\ 
m = 1827256 & HMETIS & 3,694,884,690,563 & 5733 & 6828 & 710.084 \\ 
p = 16 & HMETIS + SA & 3,681,874,557,307 & 5733 & 6830 & 710.128 \\ 
& KaHIP & 4,347,865,055,912      &    52  &      8955 &   43.9 \\ 
& KaHIP + SA & 4,309,969,305,955    &      52    &    8907  & 43.94 \\ \hline

peptide trp cage &METIS & 35,742,302,607 & 1228 & 1414 &  0.025795 \\ 
n = 16863 & METIS + SA & 35,740,265,780 & 1228 & 1414 & 0.029837 \\ 
m = 176300 & HMETIS & 35,428,817,730 & 1214 & 1472 & 31.0506 \\ 
p = 16 & HMETIS + SA & 35,237,003,004 & 1214 & 1472 & 31.0545 \\ 
& KaHIP & 43,551,196,287       & 515    &    1898   &  2.81\\ 
& KaHIP + SA & 43,388,946,192      &   536      &  1896  &  2.81 \\
\hline

urea crystal &METIS & 4,126,744,977 & 608 & 708 &  0.047032 \\ 
n = 3584 & METIS + SA & 4,126,744,977 & 608 & 708 & 0.057645 \\ 
m = 109067 & HMETIS & 5,913,680,136 & 643 & 811 & 15.2321 \\ 
p = 16 & HMETIS + SA & 5,194,749,106 & 604 & 785 & 15.2443 \\ 
& KaHIP &  3,907,671,473 & 622 & 630 &  1.05 \\ 
& KaHIP + SA & 3,907,671,473     &    622      &   630  &  1.05 \\
\hline
\end{tabular}
\end{center}
\end{table*}

\begin{figure}
  \centering
  \includegraphics[width=0.7\textwidth]{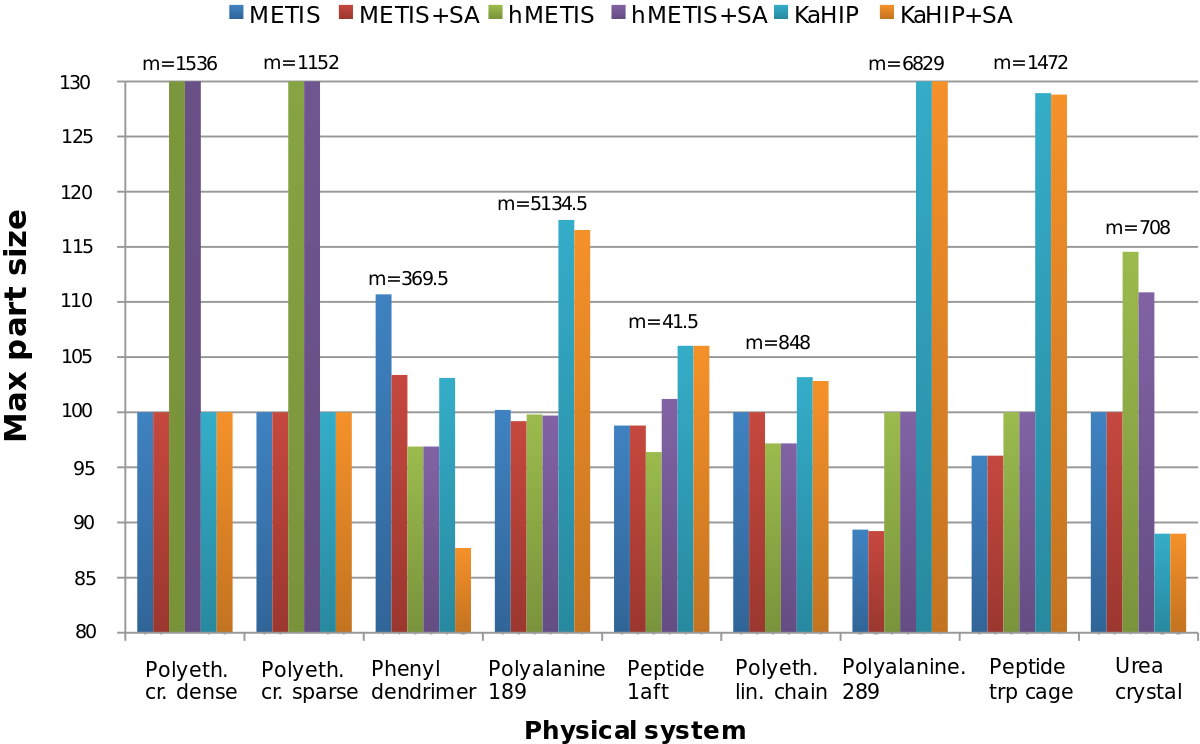}
  \caption{Performance of various methods with respect to the maximal size of the CH-partitions. \label{fig:maxChart}}
\end{figure}

\end{document}